\documentclass{amsart}

\usepackage{amsmath,amsthm,amsfonts, amssymb,amscd}
\usepackage[all]{xy}

\newtheorem{theorem}{Theorem}[section]
\newtheorem{lemma}[theorem]{Lemma}
\newtheorem{remark}[theorem]{Remark}
\newtheorem{corollary}[theorem]{Corollary}
\newtheorem{proposition}[theorem]{Proposition}

\newcommand{\dx}{\, \mbox{\rm d}}

\newcommand{\Sh}{\mbox{\rm Sh}}

\begin{document}
\title[Observables on Quantum Structures]{Observables on Quantum Structures}
\author[A. Dvure\v{c}enskij, M. Kukov\'a]{Anatolij Dvure\v censkij$^{1,2}$, M\'aria Kukov\'a$^{3}$}
\date{}
\maketitle
\begin{center}  \footnote{


The paper has been supported by Center of Excellence SAS -~Quantum
Technologies~-,  ERDF OP R\&D Project
meta-QUTE ITMS 26240120022, the grant VEGA No. 2/0059/12 SAV and by
CZ.1.07/2.3.00/20.0051. }
Mathematical Institute,  Slovak Academy of Sciences,\\
\v Stef\'anikova 49, SK-814 73 Bratislava, Slovakia\\
$^2$ Depar. Algebra  Geom.,  Palack\'{y} Univer.\\
CZ-771 46 Olomouc, Czech Republic\\
$^3$ Dep. Math., Faculty of Natural Sciences\\
Matej Bel University,
Tajovsk\'eho 40\\ SK-974 01 Bansk\'a Bystrica,
Slovakia\\
E-mail: {\tt
dvurecen@mat.savba.sk}, \ {\tt maja.kukova@gmail.com}
\end{center}

\begin{abstract} An observable on a quantum structure is any $\sigma$-homomorphism of quantum structures from the Borel $\sigma$-algebra into the quantum structure. We show that our partial information on an observable known only for all intervals of the form $(-\infty,t)$ is sufficient to determine uniquely the whole observable defined on quantum structures like $\sigma$-MV-algebras, $\sigma$-effect algebras, Boolean $\sigma$-algebras,  monotone $\sigma$-complete effect algebras with the Riesz Decomposition Property, the effect algebra of effect operators of a Hilbert space, and a system of functions, and an effect-tribe.

\end{abstract}

\section{Introduction}

Effect algebras were introduced by Foulis and Bennett \cite{FoBe} to model quantum mechanical events. This is an algebraic structure with the primary notion~- addition of mutually excluding events. A prototypical example of effect algebras is the set $\mathcal E(H)$ of all Hermitian operators of a Hilbert space $H$ which lie between the zero and the identity operators. During the last two decades effect algebras became the most important part of theory of quantum structures which studies orthomodular lattices, orthomodular posets, Boolean algebras, MV-algebras, etc. We recall that MV-algebras are algebraic semantics for many-valued logic introduced by Chang \cite{Cha}

To perform a measurement, in particular a quantum measurement, we need an analogue of a random variable, which is in our case an observable, a kind  of a $\sigma$-homomorphism of effect algebras from the Borel $\sigma$-algebra $\mathcal B(\mathbb R)$ into the given quantum structure preserving only partial addition. In quantum mechanics, an observable is simply a POV-measure (= positive operator valued measure).  Observables were studied by many authors with different aims. In the last period in the series of papers \cite{Pul, JPV, JPV1}, the authoresses concentrated to observables studied on lattice effect algebras and $\sigma$-MV-algebras exhibiting spectral properties and smearing of fuzzy observables by sharp observables using a kind of a Markov kernel. In \cite{ChKo} there were presented conditions for a functional calculus of observables in D-posets.

In the present paper, we concentrate to a question whether our information about an observable $x$ of a quantum structure concentrated only on all intervals $(-\infty, t),$ $t\in \mathbb R,$ is sufficient to derive the whole information about $x.$  We show that this is possible and we prove it for observables on $\sigma$-MV-algebras, $\sigma$-lattice effect algebras, Boolean $\sigma$-algebras, quantum logics, monotone $\sigma$-complete effect algebras with the Riesz Decomposition Property (RDP),  $\mathcal E(H),$ and  effect-tribes.   This was possible to show thanks some generalizations of the famous Loomis-Sikorski Theorem from Boolean $\sigma$-algebras to $\sigma$-MV-algebras \cite{Dvu, Mun} and to monotone $\sigma$-complete effect algebras with RDP, \cite{Dvu1}. We recall that RDP means roughly speaking a possibility to have for all two decompositions a joint refined decomposition, and thus RDP is a kind of a weak form of the distributivity.

The paper is organized as follows. In Section 2, we gather elements of theory of effect algebras, and the body of the paper is in Section 3. Section 4 indicates some possible use of observables.

\section{Elements of Effect Algebras and MV-algebras}

We recall that according to \cite{FoBe}, an {\it effect algebra} is  a partial algebra $E =
(E;+,0,1)$ with a partially defined operation $+$ and two constant
elements $0$ and $1$  such that, for all $a,b,c \in E$,
\begin{enumerate}

\item[(i)] $a+b$ is defined in $E$ if and only if $b+a$ is defined, and in
such a case $a+b = b+a;$

 \item[(ii)] $a+b$ and $(a+b)+c$ are defined if and
only if $b+c$ and $a+(b+c)$ are defined, and in such a case $(a+b)+c
= a+(b+c);$

 \item[(iii)] for any $a \in E$, there exists a unique
element $a' \in E$ such that $a+a'=1;$

 \item[(iv)] if $a+1$ is defined in $E$, then $a=0.$
\end{enumerate}

If we define $a \le b$ if and only if there exists an element $c \in
E$ such that $a+c = b$, then $\le$ is a partial ordering on $E$, and
we write $c:=b-a.$ It is clear that $a' = 1 - a$ for any $a \in E.$ As a primary source of information about effect algebras we can recommend the monograph \cite{DvPu}. It is important to note that an effect algebra is not necessarily a lattice. We recall that a {\it homomorphism}  from an effect algebra $E_1$ into another one $E_2$ is any mapping $h: E_1 \to E_2$ such that  (i) $h(1)=1$ (ii) if $a+b$ is defined in $E_1$ so is defined $h(a)+h(b)$ in $E_2$ and $h(a+b)= h(a)+h(b).$

For example, let $E$ be a system of fuzzy sets on $\Omega,$ that is $E \subseteq [0,1]^\Omega,$ such that
(i) $1 \in E$, (ii) $f \in E$ implies $1-f \in E$, and (iii) if $f,g
\in E$ and $f(\omega) \le 1 -g(\omega)$ for any $\omega \in \Omega$,
then $f+g \in E$. Then $E$ is an effect algebra of fuzzy sets which
is not necessarily a Boolean algebra. In addition, if $G$ is a partially ordered group written additively, $u \in G^+$, then $\Gamma(G,u):=[0,u]=\{g \in G: 0 \le g \le u\}$ is an effect algebra with $0=0,$ $1=u$ and $+$ is the group addition of elements if it exists in $\Gamma(G,u).$

We say that an effect algebra $E$ satisfies the Riesz Decomposition Property (RDP for short) if for all $a_1,a_2,b_1,b_2 \in E$ such that $a_1 + a_2 = b_1+b_2,$ there are four elements $c_{11},c_{12},c_{21},c_{22}$ such that $a_1 = c_{11}+c_{12},$ $a_2= c_{21}+c_{22},$ $b_1= c_{11} + c_{21}$ and $b_2= c_{12}+c_{22}.$

We define $\sum_{i=1}^n a_i:= a_1+\cdots +a_n$, if the element on the right-hand exists in $E.$ A system of elements $\{a_i: i \in I\}$ is said to be {\it summable} if, for any finite set $F$ of $I,$ the element $a_F:= \sum_{i\in F} a_i$ is defined in $E.$  If there is an element $a:= \sup \{a_F:  F$ is a finite subset of $I\},$ we call it the {\it sum} of $\{a_i: i \in I\}$ and we write $a = \sum\{a_i: i \in I\}.$

An effect algebra $E$ is {\it monotone} $\sigma$-{\it complete} if, for any sequence $a_1 \le a_2\le \cdots,$ the element $a = \bigvee_n a_n$  is defined in $E$ (we write $\{a_n\}\nearrow a$). This is equivalent that every summable sequence has a sum.

An {\it effect-tribe}  is any system ${\mathcal T}$ of fuzzy sets on
$\Omega\ne \emptyset $ such that (i) $1 \in {\mathcal T}$, (ii) if $f
\in {\mathcal T},$ then $1-f \in {\mathcal T}$, (iii) if $f,g \in {\mathcal T}$,
$f \le 1-g$, then $f+g \in {\mathcal T},$ and (iv) for any sequence
$\{f_n\}$ of elements of ${\mathcal T}$ such that $f_n \nearrow f$
(pointwise), then $f \in {\mathcal T}$. It is evident that any
effect-tribe is a monotone $\sigma$-complete effect algebra.

A very important family of effect algebras is the family of
MV-algebras, which were introduced by Chang \cite{Cha}.

We recall that an MV-algebra is an algebra $M = (M;\oplus, ^*,0,1)$
of type (2,1,0,0) such that, for all $a,b,c \in M$, we have

\begin{enumerate}
\item[(i)]  $a \oplus  b = b \oplus a$;
\item[(ii)] $(a\oplus b)\oplus c = a \oplus (b \oplus c)$;
\item[(iii)] $a\oplus 0 = a;$
\item[(iv)] $a\oplus 1= 1;$
\item[(v)] $(a^*)^* = a;$
\item[(vi)] $a\oplus a^* =1;$
\item[(vii)] $0^* = 1;$
\item[(viii)] $(a^*\oplus b)^*\oplus b=(a\oplus b^*)^*\oplus a.$
\end{enumerate}

If we define a partial operation $+$ on $M$ in such a way that $a+b$
is defined in $M$ if and only if $a \le b^*$ and  we set
$a+b:=a\oplus b$, then $(M;+,0,1)$ is an effect algebra with RDP.
For example, if $G$ is a lattice ordered group, and $u \ge0,$ then
$(\Gamma(G,u); \oplus, ^*, 0,u)$, where  $\Gamma(G,u):= \{g\in G: 0\le g \le u\},$ $a\oplus b :=(a+b) \wedge u$, and $a^* := u-a$ $(a,b \in \Gamma(G,u))$ is a prototypical example of MV-algebras. We recall that every MV-algebra is a distributive lattice.

We recall that a {\it tribe} on $\Omega \ne \emptyset$
is a collection ${\mathcal T}$ of fuzzy sets from $[0,1]^\Omega$ such
that (i) $1 \in {\mathcal T}$, (ii) if $f \in {\mathcal T}$, then $1 - f \in
{\mathcal T},$ and (iii) if $\{f_n\}$ is a sequence from ${\mathcal T}$,
then $\min \{\sum_{n=1}^\infty f_n,1 \}\in {\mathcal T}.$  A tribe is
always a $\sigma$-complete MV-algebra. For example if $f_n = \chi_{A_n},$ then  $\min \{\sum_{n=1}^\infty \chi_{A_n},1 \} = \chi_{\bigcup_n A_n}.$

\section{Observables}

Let $M$ be a $\sigma$-MV-algebra. A mapping $x: \mathcal B(\mathbb R) \to M$ is said to be an {\it observable} on $M$ if (i) $x(\mathbb R)=1,$ (ii) if $E$ and $F$ are mutually disjoint Borel sets, then $x(E \cup F)=x(E)+x(F),$ where $+$ is the partial addition on $M,$ and (iii) if $\{E_i\}$ is a sequence of Borel sets such that $E_i \subseteq E_{i+1}$ for every $i$ and $E= \bigcup_i E_i,$ then $x(E) = \bigvee_i x(E_i).$ In other words, an observable is a $\sigma$-homomorphism of effect algebras.

The definition of an observable can be literally  extended also for monotone $\sigma$-complete effect algebras.

We recall that for all $E,F \in \mathcal B(\mathbb R)$, we have (i) $x(\mathbb R \setminus E) = x(E)',$ (ii) $x(\emptyset) = 0,$ (iii)  $x(E)\le x(F)$ whenever $E \subseteq F$ and $x(F \setminus E) = x(F) - x(E)$ whenever $E \subseteq F,$ (iv) if $F_i \supseteq F_{i+1}$ and $F = \bigcap_i F_i$, then $x(F) = \bigwedge _i x(F_i).$

For example, let $\{a_n: n \in N\}$ be a finite or infinite sequence of summable elements, $\sum_{n\in N}a_n=1,$ of a monotone $\sigma$-complete effect algebra $M$ and let $\{t_n: n \in N\}$ be a sequence of mutually different real numbers. Then the mapping $x: \mathcal B(\mathbb R)\to M$ defined by

$$x(E):= \sum\{a_n: t_n \in E\}, \ E \in \mathcal B(\mathbb R). \eqno(3.0)
$$
is an observable on $M.$ In particular, if $t_0=0,$ $t_1=1$ and $a_0=a,$ $a_1=a'$ for some fixed element $a \in M,$  $x$ defined by (3.0) is an observable, called a {\it question} corresponding to the element $a.$

We denote by $\mathcal R(x):=\{x(E): \mathcal B(\mathbb R)\},$ the {\it range} of $x.$ Then the range is not necessarily  an effect subalgebra as well as not an MV-subalgebra of $M.$ For example, let $M=\{0,1/5,2/5,3/5,4/5,1\}=\Gamma(\frac{1}{5}\mathbb Z,1),$ where $\mathbb Z$ is the $\ell$-group of integers, be the MV-algebra. If $a_1 = 1/5,$  $a_2 =4/5$, and $t_1=1,$ $t_2=2,$ then the observable $x$ defined by (3.0) has the range $\mathcal R(x)=\{0,a_1,a_2,1\}$ which is not a subalgebra of $M.$

\begin{lemma}\label{le:2.1}
Let $h$ be an MV-homomorphism from an MV-algebra $M_1$ onto an MV-algebra $M_2.$ If there are two elements $A,B \in M_1,$ $A \le B,$ and an element $c \in M_2$ such that $h(A)\le c \le h(B),$ then there is an element $C \in M_1$ such that $A\le C\le B$ and $h(C) = c.$
\end{lemma}

\begin{proof}  Since $h$ is onto, there is an element $C_1\in M_1$ such that $h(C_1)=c.$ If we set $C = A\vee (B\wedge C_1),$ then $C$ is the element in question.
\end{proof}

\begin{theorem}\label{th:2.2}
Let $x$ be an observable on a $\sigma$-MV-algebra. Given a real number $t \in \mathbb R,$ we define

$$ x_t := x((-\infty, t)). \eqno(3.1)
$$
Then

$$ x_t \le x_s \quad {\rm if} \ t < s, \eqno (3.2)$$

$$\bigwedge_t x_t = 0,\quad \bigvee_t x_t =1, \eqno(3.3)
$$
and
$$ \bigvee_{t<s}x_t = x_s, \ s \in \mathbb R. \eqno(3.4)
$$

Conversely, if there is a system $\{x_t: t \in \mathbb R\}$ of elements of $M$ satisfying {\rm
(3.2)--(3.4)}, then there is a unique observable $x$ on $M$ such that $(3.1)$ holds for any $t \in \mathbb R.$
\end{theorem}

\begin{proof} Let $x$ be an observable on $M$. Let $\mathbb Q$ be the set of  rational numbers. Due to the density of $\mathbb Q$ in $\mathbb R,$ we have that (3.3) and (3.4) exist in $M,$ and equal  the corresponding suprema and infima taken for $t \in \mathbb Q.$ For example, to prove (3.4), let $\{r_n\}$ be a sequence of rational numbers such that $\{ r_n\}\nearrow s,$ then $\bigvee_n x_{r_n} = x_s.$ We recall that (3.4) holds even if $s\in \mathbb R$ and the left side supremum is taken through all rational numbers under $s.$

Conversely, let  $\{x_t: t \in \mathbb R\}$ be  a system of elements of $M$ satisfying {\rm(3.2)--(3.4)}. Let $r_1,r_2,\ldots$ be any enumeration of the set of rational numbers. Due to the Loomis-Sikorski Theorem,  \cite{Dvu1, Mun}, there are a tribe $\mathcal T$ of  functions from $[0,1]^\Omega$ for some non-void set $\Omega$ and a $\sigma$-MV-homomorphism $h$ from $\mathcal T$ onto $M.$  Given $r_n,$ let $a_n$ be a function from the tribe $\mathcal T$ such that $h(a_n)=x_{r_n}$ for any $n\ge 1.$ We are claiming that it is possible to find such a sequence of functions $\{b_n\}$ from $\mathcal T$ such that $
h(b_n)=x_{r_n}$ for any $n\ge 1$ and $b_n \le b_m$ whenever $r_n < r_m.$ Indeed, if $n=1$, we set $b_1 = a_1.$ By mathematical induction suppose that we have find $b_1,\ldots,b_n$ such that $h(b_i)=x_{r_i},$ and $b_i \le b_j$ whenever $r_i < r_j$ for $i,j=1,\ldots, n.$ Let $j_1,\ldots,j_n$ be a permutation of $1,\ldots,n$ such that $r_{j_1}<\cdots<r_{j_n}.$  For $r_{n+1}$ we have three possibilities (i) $r_{n+1}< r_{j_1},$ (ii) there exists $k =1,\ldots,n-1$ such that $r_{j_k} < r_{n+1} < r_{j_{k+1}},$ or (iii) $r_{j_n} < r_{n+1}.$  Applying Lemma \ref{le:2.1}, we can find $b_{n+1}\in \mathcal T,$ $h(b_{n+1})=r_{n+1},$ such that for all $i,j =1,\ldots, n+1,$ $b_i \le b_j$ whenever $r_i < r_j.$

Thus, we can assume that the sequence of functions $\{b_{r_n}\},$ where $b_{r_n}:= b_n,$ for $n \ge 1,$  is linearly ordered. Due to the density of rational numbers in $\mathbb R,$ for any real number $t\in \mathbb R$ we can find a function $b_t \in \mathcal T$ such that $h(b_t)=x_t.$  Indeed, if $\{p_n\} \nearrow t$ and $\{q_n\}\nearrow t$ for two sequences of rational numbers, $\{p_n\}$ and $\{q_n\},$ we can show that $\bigvee_n b_{p_n}=\bigvee_n b_{q_n}.$  Hence, $b_t := \bigvee_n b_{r_n}$ is a well-defined element of $\mathcal T$ satisfying $h(b_t)=r_t.$  In addition, the system of functions $\{b_t: t \in \mathbb R\}$ is also linearly ordered, and $b_t \le b_s$ if $t <s.$

Let $\omega \in \Omega$ be a fixed element. We define $F_\omega(t):= b_t(\omega),$ $t \in \mathbb R.$  Due to the above proved arguments, we see that $F_\omega$ is a nondecreasing, left continuous function, such that $\lim_{t \to -\infty} F_\omega(t)=0$  and $\lim_{t \to \infty} F_\omega(t)=1.$ By \cite[Thm 43.2]{Hal}, $F_\omega$ is a distribution function on $\mathbb R$ corresponding to a unique probability measure $P_\omega$ on $\mathcal B(\mathbb R),$ that is $P_\omega((-\infty,t))=F_\omega(t)$ for every $t \in \mathbb R.$ Define now a mapping $\xi: \mathcal B(\mathbb R) \to [0,1]^\Omega$ by $\xi(E)(\omega)=P_\omega(E),$ $E \in \mathcal B(\mathbb R)$, $\omega \in \Omega.$  In particular, we have $\xi((-\infty,t)) = b_t\in \mathcal T$ for any $t \in \mathbb R.$  To prove that every $\xi(E) \in \mathcal T$ for any $E \in \mathcal B(\mathbb R),$ let $\mathcal K$ be the system of all $E \in \mathcal B(\mathbb R)$ such that $\xi(E) \in \mathcal T.$ Then $\mathcal K$ is a Dynkin system, i.e. a system of subsets containing its universe, which is closed under the set theoretical complements and countable unions of of disjoint subsets. The system $\mathcal K$ contains all intervals $(-\infty, t)$ for $t \in \mathbb R,$ all intervals of the form $[a,b),$ $a\le b$, as well as all finite unions of such  disjoint intervals $\bigcup_{i=1}^n [a_i,b_i).$  Because any finite union of intervals $\bigcup_{j=1}^m [c_j,d_j)$ can be expressed as a finite union of disjoint intervals, $\mathcal K$ contains also such unions. Therefore, if $E$ and $F$ are two finite unions of intervals, so is its intersection. Hence, by \cite[Thm 2.1.10]{Dvu}, $\mathcal K$ is also a $\sigma$-algebra, and finally we have $\mathcal K = \mathcal B(\mathbb R).$

Therefore, $\xi$ is an MV-observable on $\mathcal T$ and $x:=h\circ \xi$ is an observable on $M$ such that $x((-\infty, t)) = x_t$ for any $t \in \mathbb R.$

Finally, we prove the uniqueness of $x.$  Assume, that $y$ is any observable on $M$ such that $y((-\infty,t))=x_t,$ $t \in \mathbb R.$  Let $\mathcal H$ be the system of Borel sets $E \in \mathcal B(\mathbb R)$ such that $x(E) = y(E).$ In a similar way as for $\mathcal K$, we show that $\mathcal H = \mathcal B(\mathbb R).$
\end{proof}

\begin{remark}\label{rm:2.3}
If $x$ is an observable on a monotone $\sigma$-complete effect algebra $E,$ then the first part of Theorem {\rm \ref{th:2.2}} holds.
\end{remark}

\begin{remark}\label{rm:2.4}
Theorem {\rm \ref{th:2.2}} holds also if we have a system $\{x_t: t \in \mathbb Q\}$ satisfying $(3.2)-(3.4).$
\end{remark}

Let $E$ be an effect algebra. We say that two elements $a,b \in E$ are {\it compatible}, and we write $\leftrightarrow b$,  if there are three elements $a_1, b_1,c$ such that $a= a_1+c,$ $b = b_1 +c$ and $a_1+b_1+c$ is defined in $E.$ For example, every two elements of an MV-algebra are compatible. A {\it block} is any maximal system of mutually compatible elements of $E.$ For example, if $x$ is any observable, then $x(E) \leftrightarrow x(F)$ for all $E,F \in \mathcal B(\mathbb R).$  Indeed, $x(E\cup F) = x(E\setminus F) + x(E\cap F) + x(F\setminus E).$ If $a\le b,$ then $a \leftrightarrow b$~:  $a = 0 + a$ and $b = (b-a)+a.$   Due \cite{Rie} or \cite[Thm 1.10.20]{DvPu}, if $E$ is a lattice ordered effect algebra, then each block is in fact an MV-subalgebra of $E$, and $E$ is a union of blocks. In addition, if $E$ is a $\sigma$-lattice, each  its block is a $\sigma$-MV-subalgebra of $E.$

Now we are able to extend Theorem \ref{th:2.2} for $\sigma$-complete lattice ordered effect algebras.

\begin{theorem}\label{th:2.5}  Let $E$ be a $\sigma$-lattice effect algebra. If there is a system $\{x_t: t \in \mathbb R\}$ of elements of $E$ satisfying {\rm
(3.2)--(3.4)}, then there is a unique observable $x$ on $E$ such that $(3.1)$ holds for any $t \in \mathbb R.$

\end{theorem}

\begin{proof} Since the system $\{x_t: t \in \mathbb R\}$ is linearly ordered,  we have $x_t \leftrightarrow x_s$ for all $t,s \in \mathbb R,$ and we see that there is a block $M$ of $E$ containing all $x_t$'s. This block is in fact a $\sigma$-MV-algebra.  The desired result follows now from Theorem \ref{th:2.2} for the $\sigma$-MV-algebra $M.$
\end{proof}

Now if $M$ is a Boolean $\sigma$-algebra, Theorem \ref{th:2.2} and Theorem \ref{th:2.5} can be applied for any system $\{x_t: t \in \mathbb R\}$  satisfying (3.3)-(3.4). In the following result, we give another proof to derive the corresponding observable $x$ satisfying (3.2)-(3.4). We are inspired by the proof of \cite[Thm 1.4]{Var}.

\begin{theorem}\label{th:2.6}
Let $M$ be a Boolean  $\sigma$-algebra. If there is a system $\{x_t: t \in \mathbb R\}$ of elements of $M$ satisfying {\rm
(3.2)--(3.4)}, then there is a unique observable $x$ on $M$  which is a $\sigma$-homomorphism of Boolean $\sigma$-algebras  such that $(3.1)$ holds for any $t \in \mathbb R.$
\end{theorem}

\begin{proof} Due to the classical Loomis-Sikorski Theorem, \cite[Thm 29.1]{Sik}, there are a $\sigma$-algebra $\mathcal S$ of subsets of a set $\Omega \ne \emptyset$ and a $\sigma$-homomorphism of Boolean $\sigma$-algebras $h$ from $\mathcal S$ onto $M.$ Let $r_1,r_2,\ldots$ be any enumeration of the set of rational numbers, $\mathbb Q.$ Similarly as in the proof of Theorem \ref{th:2.2}, we can find a sequence $\{A_n\}$ of elements from $\mathcal S$ such that $h(A_n) = x_{r_n}$ and $A_n \subseteq A_m$ whenever $r_n < r_m.$ If we set $A = \bigcap_n A_n,$ then $h(A)=\bigwedge_n h(A_n)=0.$  Now we define a mapping $f:\Omega \to \mathbb R$ by

$$
f(\omega)= \left\{\begin{array}{ll} \inf\{r_j: \omega \in A_j\} & \mbox{if} \ \omega \in \bigcup_n A_n,\\
0 & \mbox{if}\ \omega \notin \bigcup_n A_n.
\end{array}
\right.
$$

Then $f$ is a well-defined finite function on $\Omega.$  Moreover,
$$
f^{-1}((-\infty ,r_k)) = \left\{
\begin{array}{ll}
   \bigcup _{i:r_i < r_k}A_i &\mbox{if}\ r_k \leq 0,\\
  \bigcup _{i:r_i < r_k}A_i \cup \bigl(\Omega \smallsetminus
   \bigcup _{n=1}^{\infty} A_n\bigr) \quad &\mbox {if $r_k > 0$,}
       \end{array}
\right.
$$
hence, $f$ is $\mathcal S$-measurable. Therefore, $x:=h\circ f^{-1}$ is an observable on $M$ which is a $\sigma$-homomorphism of Boolean $\sigma$-algebras such that $x((-\infty,r_k))=h(f^{-1}((-\infty,r_k))) = x_{r_k}$ for any $k\ge 1.$ Consequently,  $x((-\infty, t))=h(f^{-1}((-\infty,t))) = x_t$ for each $t \in \mathbb R.$

Assume that $y:\mathcal B(\mathbb R) \to M$ is an arbitrary $\sigma$-homomorphism of Boolean $\sigma$-algebras such $y((-\infty,t))=x_t$ for each $t \in \mathbb R.$ Let $\mathcal K$ be the system of Borel sets $E \in \mathcal B(\mathbb R)$ such that $x(E)=y(E).$  Since both $x$ and $y$ are $\sigma$-homomorphisms of Boolean $\sigma$-algebras, trivially $\mathcal K$ is a $\sigma$-algebra containing generators $\{(-\infty, t): t \in \mathbb R\},$ and we have $\mathcal K= \mathcal B(\mathbb R).$ Hence, $x=y.$
\end{proof}

Comparing Theorem \ref{th:2.6} with Theorem \ref{th:2.2} and Theorem \ref{th:2.5}, we have the following corollary.

\begin{corollary}\label{co:2.7}
Every observable on a Boolean $\sigma$-algebra $M$ is a $\sigma$-homomorphism of Boolean $\sigma$-algebras.
\end{corollary}

\begin{proof}
Let $x$ be an observable on $M,$  and define the system $\{x_t: t\in \mathbb R\},$ where $x_t:= x((-\infty,t)).$ By Theorem \ref{th:2.6} there is a unique  $\sigma$-homomorphism $y$ of Boolean $\sigma$-algebras, $y: \mathcal B(\mathbb R)\to M,$ such that $y((-\infty,t)) =x_t$ for every $t \in \mathbb R.$ On the other hand, because $y$ is also naturally an observable on $M$, by Theorem \ref{th:2.2} or by Theorem \ref{th:2.5}, $y$ and $x$ are  unique observables on $M$ such that $x((-\infty,t))=x_t = y((-\infty,t))$ for every $t \in \mathbb R$ which entails $x=y$ and therefore, $x$ is also a $\sigma$-homomorphism of Boolean $\sigma$-algebras.
\end{proof}

Theorem \ref{th:2.6} and Corollary \ref{co:2.7} can be used also for characterizations of observables of quantum logics, which are orthomodular $\sigma$-lattices of $\sigma$-orthocomplete orthomodular posets. We recall, \cite{DvPu}, that  an {\it orthocomplementation} on a poset $L$ is a mapping $a \mapsto a^\bot$ such that for all $a,b\in L$, we have (i) $a^{\bot\bot}=a$ (ii) $b^\bot \le a^\bot$ if $a \le b,$ if $a \le b^\bot$, $a\vee b \in L,$ (iii) $a \vee a^\bot = 1.$ A poset $L$ is {\it orthomodular} if $a\le b$ implies $b = a\vee (a \vee b^\bot)^\bot$ $ (=  a \vee (b \wedge a^\bot)).$ Hence, an orthomodular poset is a special example of an effect algebra, when $a+b$ means that $a\le b^\bot$ and $a+b:= a\vee b.$  An orthomodular poset $L$ is a $\sigma$-{\it orthocomplete orthomodular poset} if, for any sequence $\{a_n\}$ such that $a_n \le a_m^\bot$ for $n\ne m,$  $\bigvee_n a_n$ exists in $L.$ Then $L$ is also a monotone $\sigma$-complete effect algebra but not necessarily with RDP, because RDP holds in our $L$ iff $L$ is a Boolean $\sigma$-algebra.

\begin{theorem}\label{th:2.8}
Let $L$ be a  $\sigma$-orthocomplete orthomodular poset. If there is a system $\{x_t: t \in \mathbb R\}$ of elements of $M$ satisfying {\rm
(3.2)--(3.4)}, then there is a unique observable $x$ on $L$  which is a $\sigma$-homomorphism of Boolean $\sigma$-algebras  such that $(3.1)$ holds for any $t \in \mathbb R.$ In addition, every observable on $L$ is a $\sigma$-homomorphism of Boolean algebras.
\end{theorem}

\begin{proof} Let $\{x_t: t\in \mathbb R\}$ be a given system of elements of $L$ satisfying (3.2)-(3.4). It is clear that $x_s \leftrightarrow x_t$ for all $t,s \in \mathbb R.$ By \cite[Thm 2.1.10]{Dvu}, a block, say $M,$ containing $\{x_t: t\in \mathbb R\}$ is a Boolean $\sigma$-algebra. Therefore, $\{x_t: t \in \mathbb R\}$ is in fact in the Boolean $\sigma$-algebra $M$ and we can apply Theorem \ref{th:2.6} and Corollary \ref{co:2.7} to obtain  the desired results.
\end{proof}


To extend Theorem \ref{th:2.2} also for monotone $\sigma$-complete effect algebras, we introduce the following notions.

We say that a {\it state} on an effect algebra $M$ is any mapping $s: M \to [0,1]$ such that (i) $s(1) =1$ and (ii) $s(a+b) = s(a) + s(b)$ whenever $a+b$ is defined in $M.$  We denote by ${\mathcal S}(M)$ the
set of all states on $M$. It can happen that ${\mathcal S}(M)$ is empty,
see e.g. \cite[Ex 4.2.4]{DvPu}. In important cases, for example when
$M$ satisfies RDP, ${\mathcal S}(M)$ is nonempty, see \cite{Rav} and \cite[Cor. 4.4]{Goo}.  In particular, every MV-algebra admits a state. We recall that  ${\mathcal S}(M)$ is always a convex set. A state $s$ is said to be {\it extremal} if $s = \lambda s_1 +
(1-\lambda)s_2$ for $\lambda \in (0,1)$ implies $s = s_1 = s_2.$
We denote by $\partial_e \mathcal S(M)$  the set of all extremal states of
${\mathcal S}(M)$. We say that a net of states, $\{s_\alpha\}$, on $M$
{\it weakly converges} to a state $s$ on $M$ if $s_\alpha(a) \to
s(a)$ for any $a \in M$. In this topology,  ${\mathcal S}(M)$ is a
compact Hausdorff topological space and every state on $M$ lies in
the weak closure of the convex hull of the extremal states as it
follows from the Krein-Mil'man theorem. Hence, $\mathcal S(M)$ is empty iff so is $\partial_e \mathcal S(M).$

If $\mathcal S(M)$ is non-void, given an element $a \in M,$ we define a function $\hat a: \mathcal S(M) \to [0,1]$ by
$$
\hat a(s):= s(a), \ s \in \mathcal S(M).
$$
Then $\hat a$ is a continuous affine function on $\mathcal S(M).$

If $M$ is an MV-algebra, $\partial_e \mathcal S(M)$ is always a compact set. In general, this is not true if $M$ is an effect algebra. However,  a fine  result of Choquet \cite[page 49]{Alf}  says that  the set of extremal states  is
always a Baire space, i.e. the Baire Category Theorem holds for $\partial_e {\mathcal
S}(M).$

Let $f$ be a real-valued function on ${\mathcal S}(M).$  We define
$$
N(f) :=\{s \in \partial_e{\mathcal S}(M) :\, f(s) \ne 0\}.
$$

Now we formulate an extension of Theorem \ref{th:2.2} and Theorem \ref{th:2.5} for monotone $\sigma$-complete effect algebras with RDP. We recall that if $M$ is a monotone $\sigma$-complete effect algebra with RDP, then $M$ is not necessarily a lattice, see \cite[Ex. 16.8]{Goo}.

\begin{theorem}\label{th:2.9}
Let $M$ be a monotone $\sigma$-complete effect algebra satisfying RDP. If there is a system $\{x_t: t \in \mathbb R\}$ of elements of $M$ satisfying {\rm
(3.2)--(3.4)}, then there is a unique observable $x$ on $M$ such that $(3.1)$ holds for any $t \in \mathbb R.$
\end{theorem}

\begin{proof} By the Loomis-Sikorski Theorem, \cite[Thm 4.1]{BCD}, we know that if $M$ is a monotone $\sigma$-complete effect algebra with RDP, then there are  a nonempty set $\Omega$, an effect-tribe $\mathcal T \subseteq [0,1]^\Omega$  with {\rm RDP}, and a $\sigma$-homomorphism $h$ from ${\mathcal T}$ onto $E$.

In the proof of the Loomis-Sikorski Theorem, \cite[Thm 4.1]{BCD}, it was used $\Omega =\mathcal S(M),$ and $\mathcal T$ was the class of all fuzzy sets $f$ on
${\mathcal S}(M)$  with the property that there exists $b \in M$ such that
$N(f-\hat b)$ is a meager subset of $\partial_e {\mathcal S}(M)$ (in the
relative topology); we write $f \sim b$. The $\sigma$-homomorphism $h$ was then defined by $h(f):=b$ iff $f \sim b.$

\vspace{2mm} {\it Claim 1.} {\it If $a\le b,$ $a,b \in M$, there are $f,g \in \mathcal T$ such that $f\le g$ and $h(f)= a,$ $h(g)=b.$}

\vspace{2mm} Indeed,  let $f \sim a$ and $g\sim b$ for some $f,g \in \mathcal T.$ We have  $\max\{f,g\} \sim a\vee b =b,$ which entails $\max\{f,g\} \in \mathcal T$ and $\max\{f,g\} \sim b.$ In a similar way, $\min\{f,g\}\in \mathcal T$ and $\min\{f,g\} \sim a.$

\vspace{2mm} {\it Claim 2.} {\it If $f,g \in \mathcal T,$ $f\le g,$ and let $c$ be an element of $M$ such that $h(f)\le c\le h(g).$  Then there exists a function $s \in \mathcal T$ such that $f\le s\le g$ and $h(s)=c.$}

\vspace{2mm}
Since $h$ is onto, there is a function  $s_1\in \mathcal T$ such that $h(s_1)=c.$ If we set $s = \max\{f, \min\{g,s_1\}\},$ by Claim 1, $s \in \mathcal T $ and    $s$ is the function in question.

\vspace{2mm}  Let again $r_1,r_2,\ldots $ be any enumeration of the set $\mathbb Q$ of rational numbers. Let $a_n \in \mathcal T$ be such a function that $a_n \sim x_{r_n}.$ Using Claim 2, similarly as in the proof of Theorem \ref{th:2.2}, we can find a sequence of functions $\{b_n\}$ from $\mathcal T$ such that $b_n \sim x_{r_n}$ and $b_n \le b_m$ whenever $r_n < r_m.$

Now the end of the present proof is identical with the end of the proof of Theorem \ref{th:2.2}.
\end{proof}

Finitely we discuss this problem also for one of the most important examples of effect algebras, the system $\mathcal E(H)$ of all Hermitian operators on a (real, complex, or quaternionic) Hilbert space $H$ which are between the zero and the identity operators.  This effect algebra is monotone $\sigma$-complete which is not a lattice and RDP fails. This example is not covered by any of the previous theorems.

\begin{theorem}\label{th:2.10}
If there is a system $\{x_t: t \in \mathbb R\}$ of elements of $\mathcal E(H)$ satisfying {\rm
(3.2)--(3.4)}, then there is a unique observable $x$ on $\mathcal E(H)$ such that $(3.1)$ holds for any $t \in \mathbb R.$
\end{theorem}

\begin{proof}  Let $\Omega(H)$ be the set of all unit vectors $\phi\in H.$ Given a Hermitian operator $A,$ let $f_A: \Omega(H) \to [0,1]$ be a function defined by $f_A(\phi):=(A\phi,\phi),$ $\phi \in \Omega(H),$ where $(\cdot,\cdot)$ is an inner product on $H.$ Then the system $\mathcal T(H)=\{f_A: A \in \mathcal E(H)\}$ is an effect-tribe which is $\sigma$-isomorphic with $\mathcal E(H)$ under the isomorphism $\imath: \phi_A \mapsto A.$

Let $\{x_t: t \in \mathbb R\}$ be a system of elements of $\mathcal E(H)$ satisfying {\rm (3.2)--(3.4)}.  Given a unit vector $\phi \in \Omega(H),$ let $F_\phi: \mathbb R \to [0,1]$ be defined by $F_\phi(t):= (x_t\phi,\phi),$ $t \in \mathbb R.$ Then $F_\phi$ is a nondecreasing, left continuous function such that $\lim_{t\to -\infty} F_\phi(t)=0$ and  $\lim_{t\to \infty} F_\phi(t)=1,$ i.e. $F_\phi$ is a distribution function corresponding to a unique probability measure $P_\phi$ on $\mathcal B(\mathbb R),$  \cite[Thm 43.2]{Hal}. Hence, we define $\xi: \mathcal B(\mathbb R) \to [0,1]^{\Omega(H)}$ by  $\xi(E)(\phi):= P_\phi(E),$ $E \in \mathcal B(\mathbb R),$ $\phi \in \Omega(H).$ To show that $\xi(E) \in \mathcal T(H),$  let $\mathcal K$ be the set of all Borel sets $E \in \mathcal B(\mathbb R)$ such that $\xi(E) \in \mathcal T(H).$  Similarly as in the proof of Theorem \ref{th:2.2}, $\mathcal K$ is a Dynkin system such that $\mathcal K = \mathcal B(\mathbb R).$ Hence, $\xi$
is an observable on $\mathcal T(H)$, and $x :=\imath \circ \xi$ is an observable on $\mathcal E(H)$ such that $x((-\infty,t))=x_t,$ $t \in \mathbb R.$

The uniqueness follows from the same arguments as the uniqueness in the proof of Theorem \ref{th:2.2}.
\end{proof}

If an effect algebra is an effect-tribe, we can also prove a characterization of an observable via the system $\{x_t: t \in \mathbb R\}.$

\begin{theorem}\label{th:2.11}
Let $\mathcal T$ be an effect-tribe on a nonempty set $\Omega.$
If there is a system $\{x_t: t \in \mathbb R\}$ of elements of $ \mathcal T$ satisfying {\rm
(3.2)--(3.4)}, then there is a unique observable $x$ on $\mathcal T$ such that $(3.1)$ holds for any $t \in \mathbb R.$
\end{theorem}

\begin{proof}  Let $\{x_t: t \in \mathbb R\}$ be a given system of functions from $\mathcal T$ satisfying (3.2)-(3.4). Let $\omega \in \Omega$ be a fixed element, and define a mapping $F_\omega: \mathbb R \to [0,1]$ via $F_\omega(t)= x_t(\omega),$  $t \in \mathbb R.$ As in the proof of Theorem \ref{th:2.10}, $F_\phi$ is a distribution function corresponding to a unique probability measure $P_\omega$ on $\mathcal B(\mathbb R).$ The mapping $x: \mathcal B(\mathbb R) \to [0,1]^\Omega$ defined by $x(E)(\omega):= P_\omega(E),$ $E \in \mathcal B(\mathbb R),$ $\omega \in \Omega,$ satisfies $x((-\infty,t))=x_t$ for any $t \in \mathbb R,$ and similarly as in the proof of Theorem \ref{th:2.10}, $x$ is a unique observable on $\mathcal T$ in question.
\end{proof}

Finally we study so-called sharp observables.
An element of an effect algebra $E$ is said to be {\it sharp} if $a \wedge a'$ exists in $E$ and $a\wedge a'=0.$ Let $\Sh(E)$ be the set of sharp elements of $E.$ Then (i) $0,1\in \Sh(E),$  (ii) if $a \in \Sh(E),$ then $a'\in \Sh(E).$ If $E$ is a lattice effect algebra, then $\Sh(E)$ is an orthomodular lattice which is a subalgebra and a sublattice of $E,$  \cite{JeRi}. If an effect algebra $E$ satisfies RDP, then by \cite[Thm 3.2]{Dvu2}, $\Sh(E)$ is even a Boolean algebra.

An observable $x$ on a monotone $\sigma$-complete effect algebra is said to be {\it sharp} if its range consists only of sharp elements, that is $\mathcal R(x) \subseteq \Sh(E).$ Sharp observables are important because as it was shown in \cite[Thm 3.4]{JPV1}, every observable on a $\sigma$-lattice effect algebra is a smearing of a sharp observable (for unexplained notions see \cite{JPV1}).  Sharp observables can be characterized as follows.

\begin{theorem}\label{th:sharp}
Let $E$ be either a $\sigma$-lattice effect algebra  or a monotone $\sigma$-complete effect algebra with RDP.
Let $\{x_t: t\in \mathbb R\}$ be a given system of elements of $E$ satisfying {\rm (3.2)-(3.4)} and let every $x_t$ be sharp. Then there is a unique observable $x$ such that $(3.1)$ holds. In addition, this observable is sharp and it preserves finite intersections.
\end{theorem}

\begin{proof} Due to \cite{JeRi} and \cite[Thm 3.2]{Dvu2}, $\Sh(E)$ is a $\sigma$-complete orthomodular lattice or a Boolean $\sigma$-algebra.  In either case,  $\{x_t: t\in \mathbb R\}$ is a system of mutually compatible elements, so that it lies in a block, $M,$ $M \subseteq \Sh(E),$ which is a Boolean $\sigma$-algebra. Applying Theorem \ref{th:2.6}, we can construct a unique observable $x$ such that (3.0) holds, and $\mathcal R(x) \subseteq M.$ This observable is clearly sharp, and
$x(F)\wedge x(G)$ exists in $E$ and $x(F\cap G)=x(F)\wedge x(F)$ for all $E,F \in \mathcal B(\mathbb R)$ because $x$ is due to Theorem \ref{th:2.6} a $\sigma$-homomorphism of Boolean $\sigma$-algebras between $\mathcal B(\mathbb R)$ and the Boolean $\sigma$-algebra $M.$
\end{proof}

\section{Miscellaneous }

We say that an observable $x$ on a monotone $\sigma$-complete effect algebra $M$ (i) is {\it bounded} if there is a compact set $C$ such that $x(C)=1$, and (ii) has the {\it Jauch-Piron Property} if $x(E)=1=x(F)$, then $x(E\cap F)=1.$ For example, if $M$ is a $\sigma$-MV-algebra, any observable $x$ of $M$ has the Jauch-Piron Property.

\begin{proposition}\label{pr:3.1}
Let an observable $x$ have the Jauch-Piron Property on a monotone $\sigma$-complete effect algebra $M.$ Then there is the least closed set $C$ of $\mathbb R$ such that $x(C)=1.$
\end{proposition}

\begin{proof} We set
$$\sigma(x)=\bigcap \{C: C\ \mbox{closed},\ x(C)=1\}.
$$
Since the natural topology of $\mathbb R$ satisfies the second countability axiom, there is a sequence of closed sets $\{C_n\}$ such that $x(C_n)=1$ for all $n,$  and $\sigma(x)=\bigcap_n C_n.$ Since $x$ satisfies the Jauch-Piron Property, we can assume that $\{C_n\} \searrow C$ and $x(\sigma(x))= \bigwedge_n x(C_n)=1.$
\end{proof}

The set $\sigma(x)$ is said to be the {\it spectrum} of $x.$ For example, for the observable $x$ from (3.0), we have $\sigma(x) =\{t_1,t_2,\ldots\}^-,$ where $^-$ denotes the closure. In particular, an observable $x$ is a question iff $\sigma(x)\subseteq \{0,1\}.$

Let $h$ be a homomorphism from an effect algebra $M_1$ into another effect algebra $M_2.$ We say that $h$ has the {\it Jauch-Piron Property}, if $h(a)=1=h(b)$, there is an element $c \in M_1$ such that $c\le a,$ $c \le b$ and $h(c)=1.$
This property is equivalent to the requirement if $h(x)=0=h(y)$, there is $z \in M_1$ such that $x,y \le z$ and $h(z)=0.$

\begin{proposition}\label{pr:3.2}  If an effect algebra $M_1$ satisfies RDP, then any homomorphism $h$ from $M_1$ into an effect algebra $M_2$ has the Jauch-Piron Property.
\end{proposition}

\begin{proof} Assume $h(a)=0=h(b).$  There are four elements $c_{11},c_{12},c_{21},c_{22}$ such that $a = c_{11}+c_{12},$ $a'= c_{21}+c_{22},$ $b= c_{11} + c_{21}$ and $b'= c_{12}+c_{22}.$ Then for the element $c = c_{11}+c_{12}+c_{21}$ we have $c\ge a,b$ and $h(c)=0.$
\end{proof}

If $s$ is a $\sigma$-additive state on a monotone $\sigma$-complete effect algebra $M,$ i.e. a state $s$ such that $\{a_n\}\nearrow a$ implies $\lim_n s(a_n)=s(a),$ and if $x$ is an observable on $M,$  then $s_x:= s\circ x$ is a probability measure on $\mathcal B(\mathbb R).$  In such a case, we can define the {\it mean value}, $Exp_s(x),$ of $x$ with respect to the $\sigma$-additive state $s$ as follows
$$
Exp_s(x) := \int_{\mathbb R} t \,\dx s_x(t),
$$
assuming that the integral exists. If $f$ is any Borel measurable function,
then $f(x):= x\circ f^{-1}$ is another observable on $M$ and then $Exp_s(f(x))=\int_{\mathbb R} f(t)\, \dx s_x(t),$ if the integral exists.
For example, we can define $x^2, x^3,x^n,$ etc.  Having this, we can calculate e.g. moments of $x,$ and we can use this approach to have some statistical characteristics of observables.

\end{document}